\documentclass{amsart}
\usepackage[utf8]{inputenc}
\usepackage[T1]{fontenc}
\usepackage{lmodern}
\usepackage{cmap}
\usepackage{microtype}
\usepackage{amsmath,amssymb,bbm,amsthm}
\usepackage[pdfborder={0 0 0}]{hyperref}
\usepackage{tikz}

\newcommand{\Mu}{\mathrm{M}}
\newcommand{\Omicron}{\mathrm{O}}
\newcommand{\cdummy}{\cdot}
\newcommand{\emdash}{---}
\newcommand{\assign}{:=}
\newcommand{\mathd}{\mathrm{d}}
\newcommand{\nosymbol}{}
\newcommand\T{\rule{0pt}{2.6ex}}
\newcommand\B{\rule[-1.2ex]{0pt}{0pt}}

\newenvironment{steps}{%
    \begin{list}{\stepcounter{algoline}\tiny{\arabic{algoline}}}{
        \setlength{\topsep}{0pt}%
        \setlength{\itemsep}{0pt}%
        \setlength{\leftmargin}{1em}%
        \setlength{\labelwidth}{1em}%
    }%
}{%
    \end{list}%
}
\newcommand\labelitem[1]{\item\addtocounter{algoline}{-1}%
    \refstepcounter{algoline}\label{#1}}

\theoremstyle{definition}
\newtheorem{proposition}{Proposition}
\newtheorem{algorithm}{Algorithm}
\newtheorem{remark}{Remark}
\newtheorem{corollary}{Corollary}
\newtheorem{theorem}{Theorem}
\newtheorem{lemma}{Lemma}

\hypersetup{
    pdftitle={A Note on the Space Complexity of Fast D-Finite Function Evaluation},
    pdfauthor={Marc Mezzarobba},
    pdfkeywords={algorithms, binary splitting, special functions, holonomic functions, guaranteed numerical evaluation, space complexity, bit burst}
}

\begin{document}

\title[Space Complexity of Fast D-Finite Function Evaluation]{A Note on the
Space Complexity\\ of Fast D-Finite Function Evaluation}

\author{Marc Mezzarobba}
\address{
    Inria, AriC, LIP (UMR 5668 CNRS-ENS Lyon-Inria-UCBL)\\
    ENS de Lyon, Lyon, France
}
\email{marc@mezzarobba.net}

\maketitle

\begin{abstract}
  We state and analyze a
  generalization of the ``truncation trick'' suggested by Gourdon and Sebah to
  improve the performance of power series evaluation by binary splitting. It
  follows from our analysis that the values of D\mbox{-}finite functions
  (i.e., functions described as solutions of linear differential equations
  with polynomial coefficients) may be computed with error bounded by $2^{-
  p}$ in time $\mathrm{O} (p (\lg p)^{3 + o (1)})$
  and space $\mathrm{O} (p)$. The standard fast algorithm for
  this task, due to Chudnovsky and Chudnovsky, achieves the same time
  complexity bound but requires $\mathrm\Theta (p \lg p)$ bits of
  memory.
\end{abstract}

\section{Introduction}\label{sec:intro}

Binary splitting is a well-known and widely applicable technique for the fast
multiple precision numerical evaluation of rational series. For any series
$\sum_n s_n$ with \smash{$\limsup_n  \left| s_n \right|^{1 / n} < 1$} whose
terms~$s_n$ obey a linear recurrence relation with polynomial coefficients,
e.g.,
\[ \ln 2 = \sum_{n = 0}^{\infty} s_n, \hspace{2em} s_n = \frac{1}{(n + 1
  ) 2^{n + 1}}, \hspace{2em} 2 (n + 2) s_{n + 1} -
  (n + 1) s_n = 0, \]
the binary splitting algorithm allows one to compute the partial sum $\sum_{n
= 0}^{N - 1} s_n$ in $\Omicron (\Mu (N (\lg N)^2
))$ bit
operations~{\cite{ChudnovskyChudnovsky1990,BrentZimmermann2010}}. Here $\Mu
(n)$ stands for the complexity of multiple precision integer
multiplication, and $\lg \nosymbol$ denotes the binary logarithm. As $N =
\Omicron (p)$ terms of the series are enough to make the
approximation error less than~$2^{- p}$, the complexity of the algorithm is
softly linear in the precision~$p$, assuming $\Mu (n) = \Omicron
(n (\lg n)^{\Omicron (1)})$.

Methods based on binary splitting tend to be favored in practice even in cases
when asymptotically faster algorithms (typically AGM
iterations~{\cite{BorweinBorwein1987}}) would apply. One high-profile example
is the computation of billions of digits of classical constants such as $\pi$,
$\zeta (3)$ or~$\gamma$. Basically all record computation in
recent years were achieved by evaluating suitable series using variants of
binary splitting~{\cite{GourdonSebah:records,Yee:records}}.

A drawback of the classical binary splitting algorithm, both from the
complexity point of view and in practice, is its comparatively large memory
usage. Indeed, the algorithm amounts to the computation of a product tree of
matrices derived from the recurrence{\emdash}see Sect.~\ref{sec:classical}
below for details. The intermediate results are matrices of rational numbers
whose bit sizes roughly double from one level to the next. Near the root,
their sizes can (and in general do) reach $\mathrm\Theta (p \lg p)$,
even though the output has size~$\mathrm\Theta (p)$.

However, the space complexity can be lowered to~$\Omicron (p)$
using a slight variation of the classical algorithm. The basic idea is to
truncate the intermediate results to a precision~$\Omicron (p)$
when they start taking up more space than the final result. Of course, these
truncations introduce errors. To make the trick into a genuine algorithm, we
need to analyze the errors, add a suitable number of ``guard digits'' at each
step and check that the space and time complexities of the resulting process
stay within the expected bounds.

The opportunity to improve the practical behavior of binary splitting using
truncations has been noticed by authors of implementations on several
occasions over the last decade or so. Gourdon and
Sebah~{\cite{GourdonSebah2001}} describe truncation as a ``crucial''
optimization. Besides the expected drop of memory usage, they report running
time improvements by an ``appreciable'' constant factor. Cheng et
al.~{\cite{ChengHanrotThomeZimaZimmermann2007}} compare truncation with
alternative (less widely applicable but sometimes more efficient) approaches.
Most recently, Kreckel~{\cite{Kreckel2008}} explicitly asks how to make sure
that the new roundoff errors do not affect the correctness of the result.

Indeed, the above-mentioned error analysis did not appear in the literature
until very recently. An article by
Yakhontov~{\cite{Yakhontov2011a,Yakhontov2011}} now provides the required
bounds in the case of the generalized hypergeometric series~$_p F_q$, which
covers all examples where the truncation trick had been used before. But the
applicability of the method is actually much wider.

The purpose of this note is to present a more general and arguably simpler
analysis. Our version is more general in two main respects.
  First, besides hypergeometric series, it applies to the solutions of linear
  ordinary differential equations with rational coefficients, also known as
  {\emph{D\mbox{-}finite}} (or holonomic) series~{\cite{Stanley1980}}.
  D\mbox{-}finite series are exactly those whose coefficients obey a linear
  recurrence relation with rational coefficients, while hypergeometric series
  correspond to recurrences of the first order.
  Second, we take into account the coefficient size of the recurrence that
  generates the series to be computed. Allowing the size of the coefficients
  to vary with the target precision~$p$ makes it possible to use the modified
  binary splitting procedure as part of the ``bit burst''
  algorithm~{\cite{ChudnovskyChudnovsky1990}} to handle evaluations at general
  real or complex points approximated by rationals of size~$\mathrm\Theta (p
 )$.

Additionally, our analysis readily adapts to other applications of binary
splitting. The simplicity and generality of the proof are direct consequences
of viewing the algorithm primarily as the computation of a product tree.
See Gosper~{\cite{Gosper1990}} and
Bernstein~{\cite[{\S}12--16]{Bernstein2008}} for further comments on this point of view.

The remainder of this note is organized as follows. Section~\ref{sec:setting}
contains some notations and assumptions. In Sect.~\ref{sec:classical}, we
recall the standard binary splitting algorithm, which will serve as a
subroutine in the linear-space version. Then, in Sect.~\ref{sec:trunc}, we
state and analyze the ``truncated'' variant that achieves the linear space
complexity for general D\mbox{-}finite functions. Finally,
Sect.~\ref{sec:final} offers a few comments on other variants of the
binary splitting method and possible extensions of the analysis.

\section{Setting}\label{sec:setting}

The performance of the binary splitting algorithm crucially depends on that of
integer multiplication. Following common usage, we denote by $\Mu (n
)$ a bound on the time needed to multiply two integers of at most
$n$~bits. Currently the best theoretical bound~{\cite{Fuerer2009}} is $\Mu
(n) = \Omicron (n (\lg n) \exp \Omicron (
\lg^{\ast} n))$, where $\lg^{\ast} n = \min \{k \lg^{\circ
k} n \leqslant 1\}$. In practice, implementations such as
GMP~{\cite{GMP}} use variants of the Sch\"onhage-Strassen algorithm of
complexity $\Omicron (n (\lg n)  (\lg \lg n)
)$. We make the usual assumption~{\cite{GathenGerhard2003}} that the
function $n \mapsto \Mu (n) / n$ is nondecreasing. It follows
that $\Mu (n) + \Mu (m) \leqslant \Mu (n + m
)$. We also assume that the {\emph{space}} complexity of integer
multiplication is linear, which is true for the standard
algorithms.

Write $\mathbbm{K}=\mathbbm{Q} (i)$, and define the {\emph{bit
size}} of a number $(x + iy) / w \in \mathbbm{K}$ (where $w, x, y
\in \mathbbm{Z}$) as $\left\lceil \lg w \right\rceil + \left\lceil \lg x
\right\rceil + \left\lceil \lg y \right\rceil + 1$.
Consider a linear differential equation with coefficients in~$\mathbbm{K}
(z)$. It will prove convenient to clear all denominators (both
polynomial and integer) and multiply the equation by a power of~$z$ to write
it as
\begin{equation}
  \Bigl(a_r (z)  \Bigl(z \frac{\mathd}{\mathd z}\Bigr)^r +
  \cdots + a_1 (z) z \frac{\mathd}{\mathd z} + a_0 (z
 )\Bigr) \cdot y (z) = 0, \hspace{2em} a_k \in
  \mathbbm{Z}[i][z]. \label{eq:deq}
\end{equation}
Let~$s = \max_k \deg a_k$, and let $h_1$ denote the maximum bit size of the
coefficients of the~$a_k$. Although our complexity estimates depend on~$r$
and~$h_1$, we do not consider more general dependencies on the equation. Thus,
the~$a_k$ are assumed to vary only in ways that can be described in terms of
these two parameters. Specifically, we assume that $s = \Omicron (1
)$ and that the coefficients of $a_k (z) / a_r (0
)$ are all restricted to some bounded domain.

We also assume that~$0$ is an ordinary (i.e. nonsingular) point
of~(\ref{eq:deq}). This implies that $a_r (0) \neq 0$ and $s
\geqslant r$. The case of {\emph{regular singular}} points (those for which we
still have $a_r (0) \neq 0$ but possibly $s < r$
{\cite[Chap.~9]{Hille1976}}) is actually
similar~{\cite{vdH2001,Mezzarobba2011}}; we focus on ordinary points to avoid cumbersome notations.

Let $\rho = \min \{\left| z \right| : a_r (z) = 0\}
\in (0, \infty]$. Then any formal series solution $y (z
) = \sum_{n \geqslant 0} y_n z^n$ of~(\ref{eq:deq}) converges on the
disk $\left| z \right| < \rho$. We select a particular solution (say, by
specifying initial values $y (0), \ldots, y^{(r - 1
)} (0)$ in some fixed, bounded domain), and an evaluation
point $\zeta \in \mathbbm{K}$ with $\left| \zeta \right| < \rho$. Let~$h_2$
denote the bit size of~$\zeta$, and let $h = h_1 + h_2$. Again, $h_2$ is
allowed to grow to infinity, but we assume that $\left| \zeta \right|$~is
bounded away from~$\rho$.

Given $p \geqslant 0$, our goal is to compute a complex number $\omega \in
\mathbbm{K}$ such that $\left| \omega - y (\zeta) \right|
\leqslant 2^{- p}$. By a classical argument, which can be reconstructed by
substituting a series with indeterminate coefficients into~(\ref{eq:deq}), the
sequence~$(y_n)$ obeys a recurrence relation of the form
\begin{equation}
  b_0 (n) y_{n + r} + b_1 (n) y_{n + r - 1} + \cdots
  + b_s (n) y_{n + r - s} = 0, \hspace{2em} b_j \in \mathbbm{K}
  \left[ n \right] \label{eq:rec} .
\end{equation}
Writing $a_k (z) = a_{k, 0} + a_{k, 1} z + \cdots + a_{k, s}
z^s$, the $b_j$ are given explicitly by
\begin{equation}
  b_j (n) = \sum_{k = 0}^r a_{k, j}  (n + r - j)^k .
  \label{eq:coeffs rec}
\end{equation}
Based on the matrix form of the recurrence~(\ref{eq:rec}), set
\begin{equation}
  B (n) = \begin{pmatrix}
    \zeta C (n) & 0\\
    R & 1
  \end{pmatrix} \in \mathbbm{K} (n)^{(s + 1)
  \times (s + 1)} \label{eq:B}
\end{equation}
where
\[ C (n) = \begin{pmatrix}
     & 1 &  & \\
     &  & \ddots & \\
     &  &  & 1\\
     - \frac{b_s (n)}{b_0 (n)} & \cdots & \cdots & -
     \frac{b_1 (n)}{b_0 (n)}
 \end{pmatrix}, \hspace{2em} R = \Bigl(~
     \underbrace{0 \hspace{1em} \ldots \hspace{1em} 0}_{s - r \text{ zeroes}}
     \hspace{1em} 1 \hspace{1em} \underbrace{0 \hspace{1em} \ldots
     \hspace{1em} 0}_{r - 1 \text{ zeroes}} ~\Bigr). \]
Let $P (a, b) = B (b - 1) \cdots B (a + 1
) B (a)$ for all $a \leqslant b$. (In particular, $P (
a, a)$ is the identity matrix.)

One may check that $b_0 (n) \neq 0$ for~$n \geqslant 0$, due to
the fact that $0$~is an ordinary point of~(\ref{eq:deq}). Thus the computation
of a partial sum $S_N = \sum_{n = 0}^{N - 1} y_n \zeta^n$ reduces to that of
the matrix product $P (0, N)$. Indeed, we have
\[ ( y_{n + r - s} \zeta^n, \dots,  y_{n + r - 1} \zeta^n, S_n)^{\operatorname
T} = P(0, n) \, (y_{r - s}, \dots, y_{r-1}, 0)^{\operatorname T} \]
where $y_{r - s} = 0, \ldots, y_{- 1} = 0, y_0, \ldots, y_{r - 1}$ are easily
determined from the initial values of the differential equation.

\section{Review of the Classical Binary Splitting
Algorithm}\label{sec:classical}

Since the entries of the matrix $B (n)$ are rational functions
of~$n$, the bit size of $P (a, b)$ grows as $\Omicron (
(b - a) \lg b)$ when $b, (b - a) \rightarrow
\infty$. This bound is sharp in the sense that it is reached for some (in
fact, most) differential equations. Computing $P (a, b)$ as $B
(b - 1) \cdot \left[ B (b - 2) \cdot \left[ \cdots B
(a) \right] \right]$ then takes time at least quadratic in $b -
a$, as can be seen from the combined size of the intermediate results. The
term ``binary splitting'' refers to the technique of reorganizing the product
into a {\emph{balanced tree}} of subproducts, using the relation {$P (a,
b) = P (m, b) \cdot P (a, m)$} with $m =
\lfloor \frac{1}{2}  (a + b) \rfloor$, and so on recursively.

A slight complication stems from the fact that removing common divisors
between the numerators and denominators of the fractions appearing in the
intermediate $P (a, b) \in \mathbbm{K}^{r \times r}$ would in
general be too expensive. Multiplying the numerators and denominators
separately and doing a single final division yields better complexity bounds.
Let
\begin{equation}
  \hat{B} (n) = b_0 (n)  \check{\zeta} B (n
 ) \in \mathbbm{Z}[i][n]^{(s + 1
 ) \times (s + 1)}, \quad \zeta = \hat{\zeta} /
  \check{\zeta} \quad (\hat{\zeta} \in \mathbbm{Z} \left[ i \right],
  \check{\zeta} \in \mathbbm{Z}) . \label{eq:hat-B}
\end{equation}
The entries of $\hat{B} (n)$ are polynomials of degree at
most~$r$ and bit size~$\Omicron (h)$. To compute~$P (a, b
)$ by binary splitting, we multiply the $\hat{B} (n)$ for
{$a \leqslant n < b$} using Algorithm~\ref{algo:BinSplit}, and then divide the
resulting matrix by its bottom right entry. The general algorithm considered
here was first published by Chudnovsky and
Chudnovsky~{\cite{ChudnovskyChudnovsky1990}}, with (up to minor details) the
analysis summarized in Prop.~\ref{prop:compl classical}. The idea of binary
splitting was known long before~{\cite{Gosper1990,Bernstein2008}}.

\begin{figure}[t]
    \hrule \hbox{}
    \begin{algorithm}
    \label{algo:BinSplit}
    $\operatorname{BinSplit} (a, b)$
    \newcounter{algoline}
    \begin{steps}
    \item If $b - a \leqslant (\text{some threshold})$
    \begin{steps}
    \item Return $\hat{B} (b - 1) \cdots \hat{B} ( a)$ where $\hat{B}$ is
defined by~(\ref{eq:hat-B})
    \end{steps}
    \item else
    \begin{steps}
        \item Return $\operatorname{BinSplit} (\lfloor \frac{a + b}{2} \rfloor,
  b) \cdot \operatorname{BinSplit} (a, \lfloor \frac{a + b}{2} \rfloor)$
    \end{steps}
    \end{steps}
    \end{algorithm}
    \hrule
\end{figure}

\begin{proposition}
  {\cite{ChudnovskyChudnovsky1990}} \label{prop:compl classical} As $b, N = b
  - a, h, r \rightarrow \infty$ with $r = \Omicron (N)$,
  Algorithm~\ref{algo:BinSplit} computes an unreduced fraction equal to $P
  (a, b)$ in
  $ \Omicron (\Mu \bigl( N (h + r \lg b) \bigr) \lg N) $ operations,
  using
  $ \Omicron \bigl( {N (h + r \lg b)} \bigr) $
  bits of memory. Assuming $\Mu (n) = n (\lg n) 
  (\lg \lg n)^{\Omicron (1)}$, both bounds are
  sharp.
\end{proposition}

\begin{proof}[sketch]
  The bit sizes of the matrices that get multiplied together at any given
  depth~$0 \leqslant \delta < \left\lceil \lg N \right\rceil$ in the recursive
  calls are at most $C 2^{- \delta} N (h + d \lg b)$ for
  some~$C$. Since there are at most $2^{\delta}$ such products and the
  multiplication function~$\Mu (\cdummy)$ was assumed to be
  subadditive, the contribution of each level is bounded by $\Mu (C
  (b - a)  (h + d \lg b))$, whence the total
  time complexity. See {\cite{ChudnovskyChudnovsky1990,Mezzarobba2011}} for
  details. The intermediate results stored or multiplied together at any stage
  of the computation are disjoint subproducts of $B (b - 1)
  \cdots B (a)$, and we assumed the space complexity of
  $n$\mbox{-}bit integer multiplication to be $\Omicron (n)$, so
  the space required by the algorithm is linear in the combined size of the~$B
  (n)$. Finally, it is not hard to construct examples of
  differential equations that reach these bounds.
\end{proof}

\begin{remark}
  The link between our setting and the more common description of the
  algorithm for hypergeometric series is as follows. In the notation of Haible
  and Pananikolaou~{\cite{HaiblePapanikolaou1997}} also used in Yakhontov's
  article, the partial sums of the hypergeometric series are related to its
  defining parameters $a, b, p, q$ by
  \[ \begin{pmatrix}
       \tilde{s} (i + 1) \vphantom{\frac{p (i)}{q
       (i)}}\\
       S (i) \vphantom{\frac{p (i)}{q (i
      )}}
     \end{pmatrix} = \begin{pmatrix}
       \frac{p (i)}{q (i)} & 0\\
       \frac{a (i)}{b (i)}  \frac{p (i
      )}{q (i)} & b (i) q (i)
     \end{pmatrix}  \begin{pmatrix}
       \tilde{s} (i) \vphantom{\frac{p (i)}{q (
       i)}}\\
       S (i - 1) \vphantom{\frac{p (i)}{q (i
      )}}
     \end{pmatrix}, \hspace{2em} \tilde{s} (i) = \frac{b
     (i)}{a (i)} s (i) . \]
  This equation becomes $\bigl(\begin{smallmatrix}
    B (i)\\
    T (0, i)
  \end{smallmatrix}\bigr) = \bigl(\begin{smallmatrix}
    b (i) p (i) & 0\\
    a (i) p (i) & b (i) q (i
   )
  \end{smallmatrix}\bigr)  \bigl(\begin{smallmatrix}
    B (i - 1) P (i - 1)\\
    T (0, i - 1)
  \end{smallmatrix}\bigr)$ upon clearing denominators. The standard recursive
  algorithm for hypergeometric series may be seen an ``inlined'' computation
  of the associated product tree. Each recursive step is equivalent to the
  computation of the matrix product $\bigl(\begin{smallmatrix}
    B_r P_r & 0\\
    T_r & B_r Q_r
  \end{smallmatrix}\bigr)  \bigl(\begin{smallmatrix}
    B_l P_l & 0\\
    T_l & B_l Q_l
  \end{smallmatrix}\bigr)$.
\end{remark}

We return to the evaluation of a D\mbox{-}finite power series within its disk
of convergence. From the differential equation~(\ref{eq:deq}), suitable
initial conditions, the evaluation point~$\zeta$ and a target precision~$p$,
one can {\emph{compute}}~{\cite{MezzarobbaSalvy2010}} a truncation order~$N$
such that $\left| S_N - y (\zeta) \right| \leqslant 2^{- p}$ and
\begin{equation}
  \left\{\begin{array}{ll}
    N \sim Kp = \bigl({\lg (\left| \zeta \right| / \rho
   )} \bigr)^{- 1} p, \quad & \text{if } \rho < \infty\\
    N = \mathrm\Theta (p / \lg p), & \text{if } \rho = \infty .
  \end{array}\right. \label{eq:N vs p}
\end{equation}
Combined with these estimates, Proposition~\ref{prop:compl classical} implies
the following.

\begin{corollary}
  \label{cor:eval}Write $\ell = h + r \lg p$. Under the assumptions of
  Proposition~\ref{prop:compl classical}, one can compute $y (\zeta
 )$ in $\Omicron (\Mu (\ell p \lg p))$ bit
  operations, using $\Omicron (\ell p)$ bits of memory. The
  complexity goes down to $\Omicron (\Mu (\ell p))$
  operations and $\Omicron (\ell p / \lg p)$ bits of memory when
  $a_r (z)$ is a constant.
\end{corollary}

This result is the basis of more general evaluation algorithms for
D\mbox{-}finite functions~{\cite{ChudnovskyChudnovsky1990}}. Indeed, binary
splitting can be used to compute the required series sums at each step when
solving a differential equation of the form~(\ref{eq:deq}) by the so-called
method of Taylor series~{\cite{Mathews2003}}. Corollary~\ref{cor:eval} thus
extends to the evaluation of~$y$ outside the disk $\left| z \right| < \rho$.
Chudnovsky and Chudnovsky further showed how to reduce the cost of evaluation
from $\mathrm\Omega (hp) = \mathrm\Omega (p^2)$ to softly linear
in~$p$ when~$h = \mathrm\Theta (p)$. This last situation is very natural
since it covers the case where the point~$\zeta$ is itself a $\Omicron (
p)$-digits approximation resulting from a previous computation. The
method, known as the {\emph{bit burst}} algorithm, consists in solving the
differential equation along a path made of approximations of~$\zeta$ of
exponentially increasing precision. Its time complexity is $\Omicron (
\Mu (p (\lg p)^2))$~{\cite{Mezzarobba2010}}.
The improvements from the next section apply to all these settings. See
also~{\cite{vdH2007c}} for an overview of more sophisticated applications.

\section{``Truncated'' Binary Splitting}\label{sec:trunc}

The superiority of binary splitting over alternatives like summing the series
in floating-point arithmetic results from the controlled growth of
intermediate results. Indeed, in the product tree computed by
Algorithm~\ref{algo:BinSplit}, the exact representations of most subproducts
$P (a, b)$ are much more compact than $\mathrm\Theta (p
)$\mbox{-}digits approximations would be. However, as already mentioned,
the bit sizes of the $P (a, b)$ also grow larger than~$p$ near
the root of the tree. The size of a subproduct appearing at depth~$\delta$ is
roughly $2^{- \delta} N (h + r \lg N)$. Assuming $N = \mathrm\Theta
(p)$, this means that the intermediate results get significantly
larger than the output in the top $\mathrm\Theta (\lg \lg p)$ levels of
the tree.

A natural remedy is to use a hybrid of binary splitting and naive summation.
More precisely, we split the full product $P (0, N)$ into $\Delta
= \mathrm\Theta (\ln N)$ subproducts of $\Omicron (p)$ bits
each, which are computed by binary splitting. The results are accumulated by
successive multiplications at precision $\Omicron (p)$.

We make use of the following notations to state and analyze the algorithm. In
Equations (\ref{eq:norm comparison}) to (\ref{eq:Trunc-mat}) below, the
coefficients of a general matrix $A \in \mathbbm{C}^{k \times k}$ are denoted
$a_{p, q} = x_{p, q} + iy_{p, q}$ ($1 \leqslant p, q \leqslant k$) with $x_{p,
q}, y_{p, q} \in \mathbbm{R}$. Let $\left\| \cdummy \right\|$ be a
submultiplicative norm on $\mathbbm{C}^{k \times k}$, and let $\beta_k > 0$ be
such that
\begin{equation}
  \left\| A \right\| \leqslant \beta_k \mathcal{N} (A),
  \hspace{2em} \mathcal{N} (A) = \max \{\left| x_{i, j}
  \right|, \left| y_{i, j} \right|\}_{1 \leqslant i, j \leqslant k} .
  \label{eq:norm comparison}
\end{equation}
For definiteness, assume for now that $\left\| \cdummy \right\| = \left\|
\cdummy \right\|_1$ is the matrix norm induced by the vector $1$-norm. (We
will discuss this choice later.) Then it holds that
\begin{equation}
  \mathcal{N} (A) \leqslant \left\| A \right\|_1 = \max_{j = 1}^k
  \sum^k_{i = 1} \left| a_{i, j} \right| \leqslant \sqrt[]{2} k\mathcal{N}
  (A) \label{eq:norm comparison bis}
\end{equation}
and
\begin{equation}
  \left\| P (a, b) \right\| \leqslant \prod_{n = a}^{b - 1}
  \left\| B (n) \right\| \leqslant \prod_{n = a}^{b - 1} \left(1
  + \left| \zeta \right| + \left| \zeta \right| \max_{k = 1}^s  \left|
  \frac{b_k (n)}{b_0 (n)} \right|\right) .
  \label{eq:norm P}
\end{equation}
Observe that, since $1$~is an eigenvalue of~$B (n)$ and the
norm~$\left\| \cdummy \right\|$ is assumed to be submultiplicative, we have
$\left\| B (n) \right\| \geqslant 1$ for all~$n$. Besides, it is
clear from~(\ref{eq:coeffs rec}) that $\left\| B (n) \right\|$ is
bounded.

Given $a \in \mathbbm{Q}$ and $\varepsilon < 1$, let
\begin{equation}
  \operatorname{Trunc} (a, \varepsilon) = \operatorname{sgn} (a) 
  \left\lfloor 2^e  \left| a \right| \right\rfloor 2^{- e}, \hspace{1em} e =
  \left\lceil \lg \varepsilon^{- 1} \right\rceil . \label{eq:Trunc-rat}
\end{equation}
We have $\left| \operatorname{Trunc} (a, \varepsilon) - a \right|
\leqslant \varepsilon$; the size of $\operatorname{Trunc} (a, \varepsilon
)$ is $\Omicron (\lg \varepsilon^{- 1})$ for bounded~$a$;
and $\operatorname{Trunc} (a, \varepsilon)$ may be computed in~$\Omicron
(\Mu (h + e))$ bit operations where $h$~is the bit
size of~$a$. We extend the definition to matrices $A \in \mathbbm{K}^{k \times
k}$ by
\begin{equation}
  \operatorname{Trunc} (A, \varepsilon) = \bigl(\operatorname{Trunc} (
  x_{p, q}, \beta_k^{- 1} \varepsilon) + i \operatorname{Trunc} (y_{p,
  q}, \beta_k^{- 1} \varepsilon)\bigr)_{1 \leqslant p, q \leqslant
  k}, \label{eq:Trunc-mat}
\end{equation}
so that again $\left\| \operatorname{Trunc} (A, \varepsilon) - A \right\|
\leqslant \varepsilon$. Note that we often write expressions of the form
$\operatorname{Trunc} (a \star b, \varepsilon)$ for some
operator~$\star$. Though this does not affect our complexity bounds, it is
usually better to compute the approximate value of $a \star b$ directly
instead of starting with an exact computation and truncating the result. See
Brent and Zimmermann~{\cite{BrentZimmermann2010}} for some relevant
algorithms.

The complete binary splitting algorithm with truncations is stated as
Algorithm~\ref{algo:TruncBinSplit}. Its key properties are summarized in the
following propositions.

\begin{figure}[t]
    \hrule \hbox{}
\begin{algorithm}
  \label{algo:TruncBinSplit}
  $\operatorname{TruncBinSplit}(p)$ \\
  \emph{The notation $X^{(q)}$, $q = 0, 1, \ldots$
  refers to a single memory location~$X$ at different points~$q$ of the
  computation.}
  \setcounter{algoline}{0}
  \begin{steps}
  \item Set $\varepsilon = 2^{- p}$
  \labelitem{step:N}Compute~$N$ such that $\left| S_N - y (\zeta) \right| \leqslant \varepsilon$ {\cite{vdH1999,MezzarobbaSalvy2010}}
  \labelitem{step:bound}Set $\Delta = \lceil \frac{N}{p}  (h + r \lg N) \rceil$, where $h$ and $r$ are given following Eq.~(\ref{eq:deq})
  \item Compute~$M$ such that $\max_{q = 0}^{\Delta - 1} \|P (\lfloor \frac{q}{\Delta} N \rfloor, \lfloor \frac{q + 1}{\Delta} N \rfloor)\| + \varepsilon \leqslant M \leqslant C^{N / \Delta}$, where~$C$ does \ not depend on $p, h, r$ [say, by approximating the right-hand side of~(\ref{eq:norm P}) from above with $\Omicron (\lg p)$ bits of precision]
  \item Initialize $\tilde{P}^{(0)} \assign \operatorname{id} \in \mathbbm{K}^{(s + 1) \times (s + 1)}$
  \labelitem{step:loop}For $q = 0, 1, \ldots, \Delta - 1$
  \begin{steps}
    \labelitem{step:BinSplit}$\hat{Q}^{} = (\hat{Q}_{i, j}) \assign \operatorname{BinSplit} (\left\lfloor \frac{q}{\Delta} N \right\rfloor, \left\lfloor \frac{q + 1}{\Delta} N \right\rfloor)$ (Algorithm~\ref{algo:BinSplit})
    \labelitem{step:Trunc-Q}$\tilde{Q}^{(q)} \assign \operatorname{Trunc} ( \hat{Q}_{s + 1, s + 1}^{- 1} \cdot \hat{Q}, \frac{1}{2 \Delta} M^{- \Delta + 1} \varepsilon)$
    \labelitem{step:accu}$\tilde{P}^{(q + 1)} \assign \operatorname{Trunc} ( \tilde{Q}^{(q)} \cdot \tilde{P}^{(q)}, \frac{1}{2 \Delta} M^{- \Delta + q + 1} \varepsilon)$
  \end{steps}
  \item Return~$\tilde{P}^{(\Delta)}$
\end{steps}
\end{algorithm}
\hrule
\end{figure}

\begin{proposition}
  The output $\tilde{P} = \operatorname{TruncBinSplit} (p)$ of
  Algorithm~\ref{algo:TruncBinSplit} is such that $\| \tilde{P} - P (0, N
 ) \| \leqslant 2^{- p}$.
\end{proposition}

\begin{proof}
  Set $P^{(q)} = P (0, \lfloor \frac{q}{\Delta} N \rfloor)$ and
  $Q^{(q)} = P (\lfloor \frac{q}{\Delta} N \rfloor, \lfloor
  \frac{q + 1}{\Delta} N \rfloor)$. Then, for $0 \leqslant q \leqslant
  \Delta$, it holds that
  \begin{equation}
    \| \tilde{P}^{(q)} - P^{(q)} \| \leqslant
    \frac{q}{\Delta}  \frac{\varepsilon}{M^{\Delta - q}} .
    \label{eq:induction}
  \end{equation}
  Indeed, this is true for $q = 0$. After Step~\ref{step:Trunc-Q} of each loop
  iteration, we have the bound $\| \tilde{Q}^{(q)} - Q^{(q
 )} \| \leqslant \frac{1}{2 \Delta} M^{- \Delta + 1} \varepsilon
  \leqslant \varepsilon$ since $\left\| B (n) \right\| \geqslant
  1$ for all~$n$. Using~(\ref{eq:induction}) and the inequality $\|
  \tilde{Q}^{(q)} \| \leqslant M$ from Step~\ref{step:bound}, it
  follows that
  \begin{align*}
    \|\tilde{Q}^{(q)} \tilde{P}^{(q)}- Q^{(q)} P^{(q)}\|{\leqslant}
    &
    \|\tilde{Q}^{(q)}-Q^{(q)}\| \|P^{(q)}\|+\|\tilde{Q}^{(q)}\| \|\tilde{P}^{(q)}-P^{(q)}\|\\
    {\leqslant} &
    {\frac{2 q+1}{2 {\Delta}}} {\frac{{\varepsilon}}{M^{{\Delta}-q-1}}}.
  \end{align*}
  After taking into account the truncation error from Step~\ref{step:accu}, we
  obtain
  \[ \| \tilde{P}^{(q + 1)} - P^{(q + 1)} \| = \|
     \tilde{P}^{(q + 1)} - Q^{(q)} P^{(q
    )} \| \leqslant \frac{q + 1}{\Delta}  \frac{\varepsilon}{M^{\Delta
     - q - 1}} \]
  which concludes the induction.
\end{proof}

\begin{proposition}
  \label{prop:compl trunc}Not counting the cost of Step~\ref{step:N},
  Algorithm~\ref{algo:TruncBinSplit} runs in time
  \begin{equation}
    \left\{\begin{array}{ll}
      \Omicron \bigl({\Mu (p) (h + r \lg p
     ) \lg p} \bigr), \quad & \text{if } \rho < \infty,\\
      \Omicron \bigl({\Mu (p)  (h + r \lg p
     )} \bigr), & \text{if } \rho = \infty,
    \end{array}\right. \label{eq:time complexity}
  \end{equation}
  as $p, h, r \rightarrow \infty$ with $r = \Omicron (\lg p)$ and
  $h = \Omicron (p)$. In both cases, it uses~$\Omicron (p
 )$ bits of memory (where the hidden constant is independent of
  $h$~and~$r$, under the same growth assumptions).
\end{proposition}

We neglect the cost of finding~$N$ to avoid a lengthy discussion of the
complexity of the corresponding bound computation algorithms. It could
actually be checked to be polynomial in $r$ and $\lg p$.

\begin{proof}
  Computing the bound~$M$ using Equation~(\ref{eq:norm P}) as suggested is
  more than enough to ensure that $\lg M = \Omicron (N / \Delta
 )$. It requires~$\Omicron (N)$ arithmetic operations on
  $\Omicron (\lg p)$-bit numbers, that is, $o (N (\lg
  N)^2)$ bit operations.
  
  By Proposition~\ref{prop:compl classical}, each of the $\Delta$~calls to
  $\operatorname{BinSplit}$ requires
  \[ \Omicron \bigl(\Mu (\tfrac{N}{\Delta}  (h + r \lg N)
    ) \lg N\bigr) = \Omicron \bigl(\Mu (p) \lg p \bigr) \]
  bit operations. The resulting matrices~$Q^{(p)}$ all have
  size~$\Omicron (p)$, hence the divisions from
  Step~\ref{step:Trunc-Q} can be done in~$\Omicron (\Mu (p)
 )$ operations using Newton's
  method~{\cite[Chap.~9]{GathenGerhard2003}}. The truncations in Steps
  \ref{step:Trunc-Q}~and~\ref{step:accu} ensure that the bit sizes of
  $\tilde{P}$~and~$\tilde{Q}$ are always at most
  \begin{equation}
    \lg \varepsilon^{- 1} + \Delta \lg M + \lg \Delta + \Omicron (1
   ) = \Omicron (p) . \label{eq:working precision}
  \end{equation}
  It follows that the matrix multiplications from Step~\ref{step:accu} take
  $\Omicron (\Mu (p))$ operations each. Summing up,
  each iteration of the loop from Step~\ref{step:loop} can be performed in
  $\Omicron (\Mu (p) \lg p)$ operations, for a total
  of $\Omicron (\Delta \Mu (p) \lg p)$.
  Equation~(\ref{eq:time complexity}) follows upon setting $N = \Omicron
  (p)$ or $N = \Omicron (p / \lg p)$ according
  to~(\ref{eq:N vs p}).
  
  The required memory comprises space for the current values of
  $\tilde{P}^{(q)}$ and $Q^{(q)}$, any temporary
  storage used by the operations from Steps
  \ref{step:BinSplit}~to~\ref{step:accu}, and an additional $\Omicron (\lg p)$
  bits to manipulate auxiliary variables such as $M$~and~$q$. We have seen
  that $\tilde{P}^{(q)}$ and $Q^{(q)}$ have bit
  size~$\Omicron (p)$. Besides, our assumption that fast integer
  multiplication could be performed in linear space implies the same property
  for division by Newton's method. Thus, Steps
  \ref{step:Trunc-Q}~and~\ref{step:accu} use $\Omicron (p)$~bits
  of auxiliary storage. Finally, again by Proposition~\ref{prop:compl
  classical}, the calls to Algorithm~\ref{algo:BinSplit} use $\Omicron (
  (N / \Delta)  (h + r \lg p)) = \Omicron (p
 )$ bits of memory.
\end{proof}

Plugging Algorithm~\ref{algo:TruncBinSplit} into the numerical evaluation
algorithms mentioned at the end of Sect.~\ref{sec:classical} yields
corresponding improvements for the evaluation of D\mbox{-}finite functions at
more general points. Table~\ref{table:summary} summarizes the complexity
bounds we obtain. The omitted proofs are direct adaptations of those that
apply without
truncations~{\cite{ChudnovskyChudnovsky1990,vdH1999,Mezzarobba2011}}. There
would be much to say on the hidden constant factors.
The main result may be stated more precisely as follows.

\begin{table}[t]
\begin{center}
  \begin{tabular}{ccccc}
    \hline
    & \T\B & Time & Space (classical) & Space (trunc.)\\
    \hline
    $\rho < \infty$\T & BinSplit & $\Omicron (\Mu (p (h + r
    \lg p) \lg p))$ & $\Omicron (p (h + r \lg
    p))$ & $\Omicron (p)$\\
    & BitBurst & $\Omicron (\Mu (p (\lg p)^2)
   )$ & $\Omicron (p \lg p)$ & $\Omicron (p
   )$\\
    $\rho = \infty$ & BinSplit & $\Omicron (\Mu (p (h + r
    \lg p)))$ & $\Omicron (p (r + h / \lg p))$ &
    $\Omicron (p)$\\
    & BitBurst & $\Omicron (\Mu (p (\lg p)^2)
   )$ & $\Omicron (p)$ & $\Omicron (p)$ \\
    \hline
  \end{tabular}
\end{center}
  \caption{\label{table:summary}Complexity of some D\mbox{-}finite function
  evaluation algorithms based on binary splitting. The rows labeled
  ``BinSplit'' summarize the cost of computing a single sum by binary
  splitting, with or without truncations. Those labeled ``BitBurst'' refer to
  the computation of $y (\zeta)$ by the ``bit burst'' method,
  using either of Algorithm~\ref{algo:BinSplit} and
  Algorithm~\ref{algo:TruncBinSplit} at each step. All entries are asymptotic
  bounds as $p, h \rightarrow \infty$ with $h = \Omicron (p)$. In
  the ``BinSplit'' case, we also let $r$~tend to infinity under the assumption
  that $r = \Omicron (\lg p)$. The whole point of the ``bit
  burst'' method is to get rid the dependency on~$h$.}
\end{table}

\begin{theorem}
  \label{prop:trunc-ancont}Let $U \subset \mathbbm{C}$ be a simply connected
  domain such that $0 \in U$ and $a_r (z) \neq 0$ for all $z \in
  U$. Fix $\ell_0, \ldots, \ell_{r - 1} \in \mathbbm{C}$ and $\zeta \in U$.
  Assume that $0$ is an ordinary point of (\ref{eq:deq}), and let~$y$ be the
  unique solution of~(\ref{eq:deq}) defined on~$U$ and such that $y^{(k
 )} (0) = \ell_k$, $0 \leqslant k < r$. Then, the value $y
  (\zeta)$ may be computed with error bounded by $2^{- p}$ in
  time $\Omicron (\Mu (p)  (\lg p)^2)$
  and space $\Omicron (p)$, not counting the resources needed to
  approximate the $\ell_k$ or $\zeta$ to precision $\Omicron (p)$ or to find
  suitable truncation orders for the Taylor series involved.
\end{theorem}

Finally, some comments are in order regarding the ``working precision'', that
is, the size~$p'$ of the entries of $\tilde{P}$ and $\tilde{Q}$ in
Algorithm~\ref{algo:TruncBinSplit}. Equation~(\ref{eq:working precision})
suggests a number of ``guard digits'' $p' - p = \mathrm\Theta (p)$.
Moreover, if the bound~$M$ is computed using~(\ref{eq:norm P}), the hidden
constant depends on the choice of~$\left\| \cdummy \right\|$.

Let $B_{\infty} = \lim_{n \rightarrow \infty} B (n)$. For the
norm~$\left\| \cdummy \right\|_{\mathrm{opt}}$ given by Lemma~\ref{lemma:norm}
below, we have
\[ \lg \left\| P (a, b) \right\|_{\mathrm{opt}} \leqslant \sum_{n =
   a}^{b - 1} \lg \bigl(\left\| B_{\infty} \right\|_{\mathrm{opt}} + \Omicron
   (n^{-1})\bigr) = \Omicron \bigl(\lg (b - a)\bigr), \]
and hence $\lg \left\| P (a, b) \right\| = \Omicron (\lg
(b - a))$ for any norm~$\left\| \cdummy \right\|$.

\begin{lemma}
  \label{lemma:norm}There exists a matrix norm~$\left\| \cdummy
  \right\|_{\mathrm{opt}}$ such that $\|B_{\infty} \|_{\mathrm{opt}} = 1.$
\end{lemma}

\begin{proof}
  We mimic the classical proof of Householder's
  theorem~{\cite[Sect.~4.2]{Serre2002}}. By~(\ref{eq:coeffs rec}), the limit
  $C_{\infty} = \lim_{n \rightarrow \infty} C (n)$ is the
  companion matrix of the polynomial $z^s a_r (1 / z)$. The
  eigenvalues of $\zeta C_{\infty}$ are strictly smaller than~$1$ in absolute
  value since $\left| \zeta \right| < \rho$. Let $\Gamma$ be such that
  $\Gamma^{- 1} C_{\infty} \Gamma$ is in (lower) Jordan normal form.
  Let $\lambda > 0$, and set $\Pi = \operatorname{diag} (\Gamma, 1)
  \cdot \operatorname{diag} (1, \lambda, \ldots, \lambda^s)$. Then $\Pi^{-
  1} B_{\infty} \Pi$ is lower triangular, with off-diagonal entries tending to
  zero as $\lambda \to 0$. Hence we have $\| \Pi^{- 1} B_{\infty}
  \Pi \|_1 = 1$ for $\lambda$ small enough. We choose such a~$\lambda$
  (e.g., $\lambda = \frac{1 - | \zeta | / \rho}{2 \max (1,
  | \zeta |)}$) and set $\| A \|_{\mathrm{opt}} =
  \| \Pi^{- 1} A \Pi \|_1$.
\end{proof}

One way to eliminate the overestimation in the algorithm is to compute
approximations of the matrices $P (\lfloor \frac{q}{\Delta} N \rfloor,
\lfloor \frac{q + 1}{\Delta} N \rfloor)$ with $\Omicron (\lg p
)$ digits of precision before doing the computation at full precision.
One then uses the norms of these approximate products instead of those of the
individual~$B (n)$ to determine~$M$. We can also explicitly
construct an approximation $\tilde{\Pi}$ of the matrix~$\Pi$ from the proof of
Lemma~\ref{lemma:norm} precise enough that $\| \tilde{\Pi}^{- 1} B_{\infty} 
\tilde{\Pi} \|_1 = 1$, and use the corresponding norm instead of $\left\|
\cdummy \right\|_1$ in~(\ref{eq:norm P}).
(Compare~{\cite[Algorithm~B]{vdH1999}}.) Other options include computing
symbolic bounds on the coefficients of ${P (a, b)}$ as a function
of $a$~and~$b$~{\cite{MezzarobbaSalvy2010}} or finding an explicit
integer~$n_0$ such that $n \geqslant n_0 \Rightarrow \left\| B (n
) \right\|_{\mathrm{opt}} = 1$ based on the symbolic expression of~$n$.
Which variant to use in practice depends on the features of the implementation
platform.

In any case, replacing the $\Omicron (\cdummy)$ in the space
complexity bound by an explicit constant would also require more specific
assumptions on the memory representation of the objects we work with, as well
as finer control on the space complexity of integer multiplication and
division (see, e.g., Roche~{\cite{Roche2011}}).

\section{Final Remarks}\label{sec:final}

\subsection*{What we lose and what we retain.}The price we pay for the reduced memory usage is
the ability to easily extend the computation to higher precision. Indeed, the
classical algorithm computes the exact value of the matrix~$P (0, N
)$, from which we can deduce $P (0, N')$ for any $N' > N$
in time roughly proportional to $N' - N$. This is no longer true with the
linear-space variant. In some ``lucky'' cases where~$P (0, N)$
can be represented exactly in linear space, it is possible to get the memory
usage down to~$\Omicron (N)$ while preserving restartability: see
Cheng et al.~{\cite{ChengHanrotThomeZimaZimmermann2007}} and the references
therein. Additionally, the resulting running time is reportedly lower than
using truncations, probably owing to the fact that the size of the subproducts
in the $\lg (N / \Delta)$ lower levels of the tree is reduced as
well. Unfortunately, the applicability of the technique is limited to very
special cases.

Two other traditional selling points of the
binary splitting method are its easy parallelization and good memory locality.
Nothing is lost in this respect, except that the memory bound grows to $\mathrm\Theta
(t \cdot p)$ when using $t = o (\lg N)$ parallel
tasks in the approximate part of the computation.

\subsection*{Generalizations.}The idea of binary splitting ``with
truncations'' and the outline of its analysis adapt to various settings not
covered here. For instance, we may consider systems of linear differential
equations instead of scalar equations~{\cite{ChudnovskyChudnovsky1990}}.
Product trees of matrices over number fields $\mathbbm{K}' =\mathbbm{Q} (
\alpha)$ other than $\mathbbm{Q} (i)$ or over rings of
truncated power series $\mathbbm{K}' \left[ \left[ \varepsilon \right] \right]
/ \left\langle \varepsilon^k \right\rangle$ are also useful, respectively, to
evaluate limits of D\mbox{-}finite functions at regular singular points of
their defining equations, and to make the analytic continuation process more
efficient for equations of large order~{\cite{vdH1999,Mezzarobba2011}}. It is
not essential either that the coefficients of the recurrence relation
satisfied by the~$y_n$ are rational functions of~$n$: all we really ask is
that they have suitable growth properties and can be computed fast.

\subsection*{Implementation.}We are working on an implementation of the
algorithm from Sect.~\ref{sec:trunc} in an experimental branch of the software
package NumGfun~{\cite{Mezzarobba2010}}. The current state of the code is
available from
\begin{center}
    \url{http://marc.mezzarobba.net/supplementary-material/trunc-CASC2012/}.
\end{center}
A comparison (updated periodically) with the implementation of binary splitting
without truncations used in previous releases of NumGfun is also included.

\subsection*{Acknowledgments.}I would like to thank Nicolas Brisebarre and
Bruno Salvy for encouraging me to write this note and offering useful
comments, and Anne Vaugon for proofreading parts of it.

\begin{tikzpicture}[remember picture,overlay]
\node [yshift=4.3cm, anchor=north] at (current page.south)
[text width=\textwidth, font=\scriptsize, align=justify] {
    This work is in the public domain. As such, it is not
    subject to copyright. Where it is not legally possible to consider this
    work as released into the public domain, any entity is granted the right
    to use this work for any purpose, without any conditions, unless such
    conditions are required by law. \\[.2\baselineskip]
    The present version is identical except for minor corrections and
    formatting differences to the article published in
    V.P. Gerdt et al. (Eds.), CASC 2012, LNCS 7442, pp. 212–223, 2012.
    The \href{http://www.springerlink.com/content/qr7885031g607253/}
    {original version} is available at www.springerlink.com.
};
\end{tikzpicture}

\end{document}